\documentclass[preprint,11pt]{elsarticle}
\usepackage{amsmath}
\usepackage{amssymb}
\usepackage{latexsym}
\usepackage{graphicx}

\usepackage[usenames]{color}

\input amssym.def
\newsymbol\rtimes 226F
\newfont{\nset}{msbm10}

\def\Pcum{{\cal{P}}_{\mbox{\scriptsize cum}}}

\newtheorem{theo}{Theorem}[section]

\newtheorem{proposition}[theo]{Proposition}
\newtheorem{lemma}[theo]{Lemma}
\newtheorem{definition}[theo]{Definition}

\journal{ Theor. Comput. Sci.  Published in  vol.  412 (2011) pp. 865-875. }

\renewcommand*{\today}{}

\begin{document}

\begin{frontmatter}

\title{Farey Graphs as  Models for Complex Networks
}

\author{Zhongzhi Zhang}
\address{School of Computer Science and  Shanghai Key Lab of Intelligent Information Processing,
      Fudan  University, Shanghai 200433, China  ({\tt zhangzz@fudan.edu.cn}).}
\author{Francesc Comellas}
\address{Dep. Matem\`atica Aplicada IV, EPSC,
     Universitat Polit\`ecnica de Catalunya, c/ Esteve Terradas 5, Castelldefels (Barcelona), Catalonia,
     Spain ({\tt comellas@ma4.upc.edu}).}


\begin{abstract}
Farey sequences of irreducible fractions between 0 and 1 can be related to graph constructions known as Farey graphs.
These graphs were first introduced by Matula and Kornerup in 1979 and further studied by Colbourn in 1982 and they have many interesting properties: they are  minimally 3-colorable, uniquely Hamiltonian, maximally outerplanar and perfect.
In this paper we introduce a simple generation method for a  Farey graph family, and we study analytically relevant topological properties: order, size, degree distribution and correlation, clustering, transitivity, diameter and average distance. We show that the graphs are a good model for networks associated with some complex systems.
\end{abstract}



\end{frontmatter}


\section{Introduction}
A Farey sequence of order $n$ is the  sorted sequence of irreducible fractions between $0$ and $1$ with denominators less than or equal to $n$ and arranged in increasing values. Therefore, each Farey sequence starts with $0$ and ends with $1$, denoted by $\frac{0}{1}$ and $\frac{1}{1}$, respectively.
The Farey sequences of orders 1 to 4 are:  $F_1=\{{0\over 1},{1\over 1}\}$, $F_2=\{{0\over 1},{1\over 2},{1\over 1}\}$, $F_3=\{{0\over 1},{1\over 3},{1\over 2},{2\over 3},{1\over 1}\}$, $F_4=\{{0\over 1},{1\over 4},{1\over 3},{1\over 2},{2\over 3},{3\over 4},{1\over 1}\}$.

Farey sequences, which in some papers are incorrectly called Farey
series, can be constructed using mediants (the mediant of $a\over b$
and $c\over d$ is ${a+c}\over {b+d}$):  the Farey sequence of order
$n$ is obtained from the Farey sequence of order $n-1$ by computing
the mediant of each two consecutive values in the Farey sequence of
order $n-1$, keeping only the subset of mediants that have
denominator  $n$, and placing each mediant between the two values
from which it was computed. Note that neighboring fractions in a
sequence are unimodular, i.e.  if $p/q$ and $r/s$ are neighboring
fractions, then $rq- ps = 1$. It was John Farey who  in 1816
conjectured that new terms in $F_n$ could be  obtained  as  mediants
from two consecutive terms in $F_{n-1}$. Cauchy  proved the
conjecture and used the term Farey sequences for the first time.
However, Farey sequences were in fact introduced in 1802 by C.
Haros, see \cite{HaWr79}.

Farey sequences have many interesting properties, which we will not
review here  and we refer the interested reader to the abundant
literature on this topic, see~\cite{Weweb} and references therein.
As we are interested in some connections of these sequences with
graph theory, we first mention their relation with Farey trees
(which some authors call Farey graphs, but are a different structure
than the graphs studied in this paper). A Farey tree is a binary
tree labeled in terms of a Farey sequence and it is constructed as
follows: The left child of any number is its mediant with the
nearest smaller ancestor, and the right child is the mediant  with
its nearest larger ancestor. Using 2/3 as an example, its closest
smaller ancestor is 1/2, so the left child is 3/5, and its closest
larger ancestor is 1/1, and the right child is 3/4. The process can
continue indefinitely, see Fig.~\ref{fig:fareytree}. Note that on
each level the numbers appear always in order and that all the
rationals within the interval [0,1] are included in the infinite
Farey tree. Moreover, the Farey sequence of order $n$ may be found
by an inorder traversal of this tree, backtracking whenever a number
with denominator greater than $n$ is reached. The Farey tree is a
subtree of the Stern-Brocot tree which contains all positive
rationals, see for example~\cite{GrKnPa89,Au08}.

\begin{figure}
\begin{center}
\includegraphics[width=10cm]{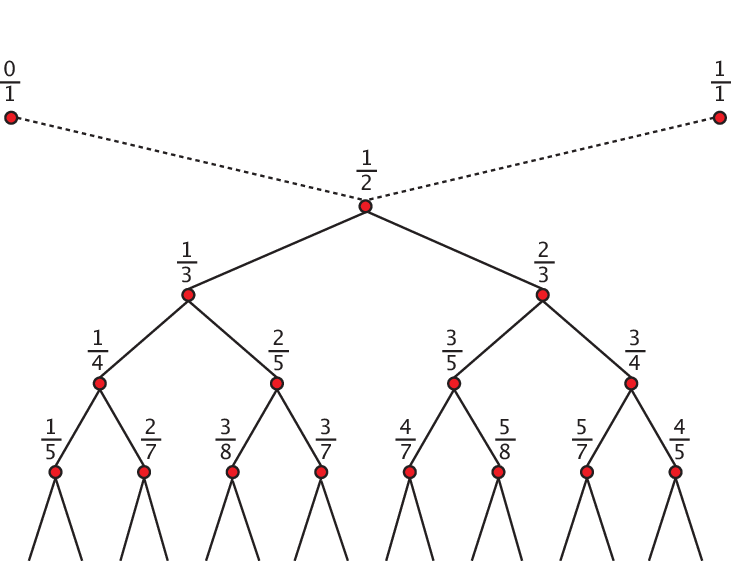}
\caption{Farey tree. }
\label{fig:fareytree}
\end{center}
\end{figure}

A Farey sequence can be related to a graph construction known as
Farey graph. A Farey graph ${\cal F}$ is a graph with vertex set on
irreducible rational numbers between $0$ and $1$, and two rational
numbers $p/q$  and $r/s$  are adjacent in ${\cal F}$  if and only if
$rq- ps = 1$ or $-1$, or equivalently if they are consecutive terms
in some Farey sequence $F_m$. Note that the graph can be obtained
from a subtree of the Farey tree by adding new edges, or
equivalently, a subtree of the Farey tree is a spanning tree of a
Farey graph. This graph was first introduced by Matula and Kornerup
in 1979 and further studied by Colbourn in 1982, and has many
interesting properties. For example,  they are  minimally
3-colorable, uniquely Hamiltonian, maximally outerplanar and
perfect, see \cite{MaKo79,Co82,Bi88}

\begin{figure}
\begin{center}
\includegraphics[width=10cm]{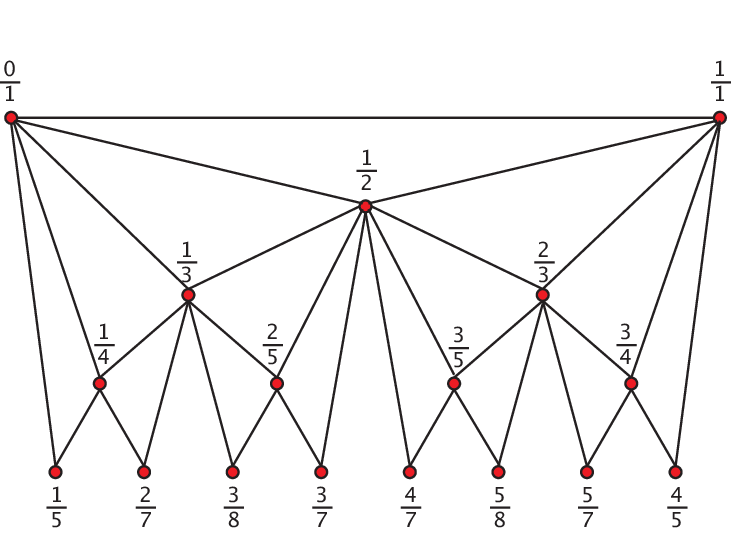}
\caption{A Farey graph with 17 vertices. }
\label{fig:fareygraph}
\end{center}
\end{figure}

In this paper we introduce a simple construction method for a family
of Farey graphs, inspired by the mediant calculation of new nodes in
the Farey tree. Other than the properties proved for general Farey
graphs in~\cite{MaKo79,Co82}, we determine analytically, for this
family of graphs, their order, size, degree distribution, degree correlations, clustering
and transitivity coefficients, diameter, and average distance. 
The graphs are of interest as models for complex systems~\cite{Ne03}, as
the parameters computed match those of their associated networks. 
They have small-world characteristics (a large clustering with  small average
distance) and they are minors of the pseudo-fractal
networks~\cite{DoGoMe02} and Apollonian graphs~\cite{AnHeAnSi05},
but in these cases the graphs are also scale-free (their degree
distribution follows a power-law), see~\cite{BaAl99}, while in our
case the degrees follow an exponential distribution. 
However, relevant networks, which describe technological and biological
systems, like  some electronic circuits and protein networks are
almost planar and have an exponential degree
distribution~\cite{BaWe00,FeJaSo01,Ne03}. 
This Farey graph family is also related to some random networks 
constructed following the
method known as geographical attachment~\cite{OzHuOt04,ZhZhWaSh07}.

\section{Definition, order and size of the Farey graphs  ${\cal F}(t)$}

In this section we give an iterative construction method for a
family of Farey graphs. 
When modeling real world networks with graphs, different methods have been considered:  
edge reconnection, duplication and addition of substructures, etc. 
Iterative methods that add new vertices at each step are useful as they can 
mimic  processes that drive the network evolution through time. 
For example, in social, collaborative 
and some technological and biological networks it is very likely
that a new node will join to nodes that are already
adjacent. This suggests the following graph construction:
\smallskip

\begin{definition}
\label{def:fareygraph} The graph ${\cal F}(t)=(V(t),E(t))$, $t\geq
0$, with vertex set $V(t)$ and edge set $E(t)$ is constructed as
follows:

For $t=0$, ${\cal F}(0)$ has two vertices and an edge joining them.

For $t\geq 1$, ${\cal F}(t)$ is obtained from ${\cal F}(t-1)$ by
adding  to every edge introduced at step $t-1$ a new vertex adjacent
to the endvertices of this edge.
\end{definition}.

Therefore, at $t=0$, ${\cal F}(0)$ is  $K_2$, at $t=1$ the graph is
$K_3$, at $t=2$ the graph has five vertices and seven edges, etc.
Notice that the graph ${\cal F}(t)$, $t>1$,  can also be constructed
recursively from two copies of  ${\cal F}(t-1)$, by identifying two
initial vertices -one from each copy of  ${\cal F}(t-1)$- and adding
a new edge between the other two initial vertices, see Fig.
\ref{fig:fareygraphs0to3}.

In what follows we will call a {\em generating edge} an edge that,
according to the definition~\ref{def:fareygraph}, is used to
introduce a new vertex in the next iteration step.

\begin{figure}
\begin{center}
\includegraphics[width=12cm]{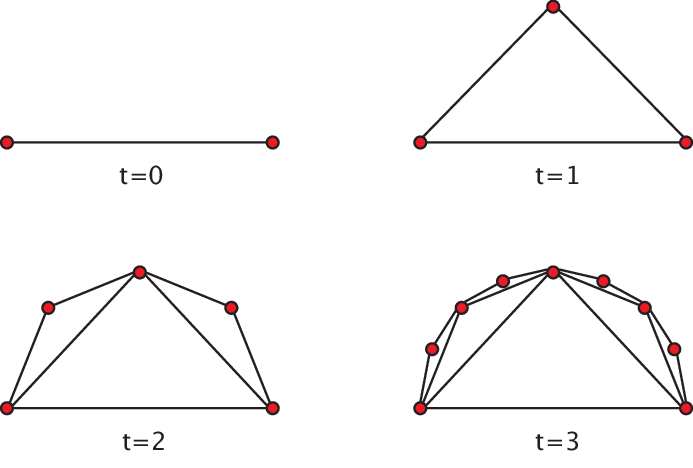}
\caption{Farey graphs ${\cal F}(0)$ to ${\cal F}(3)$. } \label{fig:fareygraphs0to3}
\end{center}
\end{figure}

This graph construction is deterministic, and uses an iteration
process similar to that of~\cite{DoGoMe02} where  Dorogovtsev
\emph{et al.} introduced a graph, which they called ``pseudofractal
scale-free web'' constructed as follows:  At each  step, for every
edge of the graph (not only those introduced at the last step as in
our graph construction), a new node is added, which is attached to
the endvertices of the edge. In their construction the starting
graph is $K_3$. This graph construction was generalized
in~\cite{CoFeRa04}. All these graphs have relevant distinct
properties with respect to the Farey graph family defined here.
Finally, our graphs constitute the extreme case $q=0$ of the random
construction in~\cite{ZhZhWaSh07}, where at each step an edge is
chosen with probability $q$, and after the insertion of the new
vertex and edges, the edge is removed.

We can see that our construction produces Farey graphs by labeling
the vertices: if the two initial vertices are labeled $0/1$ and
$1/1$, and each new added vertex is labeled with the mediant  of the
two vertices where it is joined, the vertices verify the definition
of Farey graphs as given, for example, by Colbourn in~\cite{Co82} or
Biggs in~\cite{Bi88}. Therefore, as has been proved there, the
graphs ${\cal F}(t)$ are minimally 3-colorable, uniquely
Hamiltonian, maximally outerplanar and perfect.


Thanks to the deterministic nature of the graphs ${\cal F}(t)$,
we can give exact values for  the relevant topological properties
of this graph family, namely, order, size,
degree distribution, clustering, transitivity, diameter and average distance.

To find the order and size of ${\cal F}(t)$,  we denote the number
of new vertices and edges added at step $t$ by $L_V(t)$ and
$L_E(t)$, respectively. These edges are generating edges.

Thus, initially ($t=0$), we have $L_V(0)=2$ vertices and $L_E(0)=1$
edges in ${\cal F}(0)$.

As each generating edge produces a new vertex and two generating
edges at the next iteration, we have that $L_V(t)=L_E(t-1)$ and
$L_V(t)=2\, L_E(t-1)$, which leads to $L_E(t)= 2^{t}$ and
$L_V(t)=2^{t-1}$

Therefore, the order
of the graph is $|V(t)|=\sum_{i=0}^{t}L_V(i)$ and
the total number of edges is $|E(t)|=\sum_{i=0}^{t}L_E(i)$ and we have:

\begin{proposition}
The order  and size  of the graph ${\cal F}(t)=(V(t),E(t))$ are,
respectively,
\begin{equation}\label{ordersize}
|V(t)|=2^{t}+1 \quad and \quad |E(t)|=2^{t+1}-1.
\end{equation}
\end{proposition}
$\Box$

The average degree is $\bar{\delta(t)}=4-{3}/{(2^{t}+1)}$.
For large $t$, it is small and approximately equal to $4$.

Many real-life networks are sparse in the sense that the number
of links in the network is much less than $|V(t)|(|V(t)|-1)/2$, the
maximum number of links~\cite{Ne00,AlBa02,DoMe02,Ne03,BoLaMoChHw06}.


\section{Relevant characteristics of ${\cal F}(t)$}

In this section we find analytically the degree distribution, degree correlations,
clustering and transitivity coefficients, diameter and average
distance of the graphs ${\cal F}(t)$.


\subsection{\em Degree distribution}\label{degdist}

When studying networks associated with complex systems, the degree
distribution is an important characteristic related to their
topological, functional and  dynamical properties. Most real life
networks follow a power-law degree distribution and are called
scale-free networks. However, relevant networks, which describe
technological and biological systems, like  some electronic circuits
and protein networks have an exponential degree
distribution~\cite{BaWe00,FeJaSo01,Ne03}. The well known
Watts-Strogatz small world network model also follows an exponential
degree distribution~\cite{WaSt98} as it is the case of the Farey
graphs analyzed here.

The degree distribution of ${\cal F}(t)$ is deduced from the following
facts: Initially, at $t=0$, the graph has two vertices of degree
one. When a new vertex $v$ is added to the graph at step $t_{c,v}$,
this vertex has degree $2$ and it is connected to two generating
edges. From the construction process, all vertices of the graph,
except the initial two vertices, are always connected to two
generating edges and will increase their degrees by two units at the
next step.

\medskip
\begin{proposition} \label{deg-dist}
The cumulative degree distribution of the graph ${\cal F}(t)$ follows an
exponential  distribution  $\Pcum (\delta)\sim 2^{-\frac{\delta}{2}}$
\end{proposition}

\begin{proof}
We denote the degree of vertex $v$  at step $t$ by $\delta_{v}(t)$.
By construction, we have
\begin{eqnarray}
\delta_{v}(t+1)& = &\delta_{v}(t)+2\quad  v\in V(t), v\neq {\small {0\over 1},{1\over 1}} \nonumber \\
\delta_{0\over 1}(t)&=&\delta_{1\over 1}(t)=t+1
\end{eqnarray}
if $t_{c,v}$ ($t_{c,v}>0$) is the step at which a vertex $v$  is added to the graph, then
$\delta_{v}(t_{c,v})=2$ and hence
\begin{equation}\label{deg-growing}
\delta_{v}(t)=2(t-t_{c,v}+1).
\end{equation}
Therefore, the degree distribution of the vertices of the graph
${\cal F}(t)$ is as follows:
 the number of vertices of degree
$2\cdot 1,2\cdot 2,2\cdot 3,\cdots,2\cdot t$, equals, respectively,
to $2^{t-1},  2^{t-2}, \ldots, 2, 1$ and the two initial vertices
have degree $t+1$.

The degree distribution ${\cal{P}}(\delta)$ for a network gives the
probability that a randomly selected vertex has exactly $\delta$
edges. In the analysis of the degree distribution of real life
networks,  see~\cite{JuKiKa02,Ne03}, it is usual to consider  their
cumulative degree distribution,
$$\Pcum(\delta)=\sum_{\delta'=\delta}^{\infty}{\cal{P}}(\delta'),$$
which is the probability that the degree of a vertex is greater than
or equal to $\delta$. 

Networks whose degree distributions are exponential:
${\cal{P}}(\delta) \sim e^{-\alpha\delta}$, have also an exponential
cumulative distribution with the same exponent:
\begin{equation} \label{cum-deg-dist}
\Pcum(\delta)=\sum_{\delta'=\delta}^{\infty}P(\delta')\approx
\sum_{\delta'=\delta}^{\infty}e^{-\alpha\delta'}=\left(\frac{e^{\alpha}}{e^{\alpha}-1}\right)
e^{-\alpha\delta}.
\end{equation}

Because the contribution to the degree distribution of the two
initial vertices of a Farey graph is small, we can use
Eq.~(\ref{deg-growing}),
and we have for ${\cal F}(t)$  that   
$\Pcum(\delta)\approx {\cal{P}}\left (t'\leq
\tau=t-\frac{\delta-2}{2}\right)$. Hence, for large $t$
\begin{eqnarray}\label{cumulative distribution}
\Pcum(\delta) \approx\sum_{t'=0}^{\tau}
\frac{L_V(t')}{|V(t)|}=\frac{2}{2^{t}+1}+
\sum_{t'=1}^{\tau}\frac{2^{t'-1}}{2^{t}+1} =\frac{1+2^\tau}{1+2^{t}} \sim
2^{-\frac{\delta}{2}}.
\end{eqnarray}
\end{proof}

\subsection{\em Degree correlations}

The study of degree correlations in a graph is a particularly 
interesting subject in the
context of network analysis, as they account
for some important network structure-related effects~\cite{MsSn02,Ne02,Ne03}. 
One first parameter  is the average degree of the
adjacent vertices for vertices with degree $\delta$ as  a function of this degree value~\cite{PaVaVe01}, which we denote as $k_{\rm nn}(\delta)$. 
When $k_{\rm nn}(\delta)$ increases with
$\delta$, it means that vertices have a tendency to connect to
vertices with a similar or larger degree. 
In this case the graph is called  assortative~\cite{Ne02}. 
In contrast, if $k_{\rm nn}(\delta)$ decreases with $\delta$, 
which implies that
vertices of large degree are likely to be adjacent to vertices with
small degree, then the graph is said to be disassortative. 
The graph is uncorrelated if $k_{\rm nn}(\delta)={\rm const}$.

We can obtain an exact analytical expression of $k_{\rm nn}(\delta)$ 
for the Farey graphs ${\cal F}(t)$.
Note that except for the initial two vertices of step $0$,
no vertices introduced at the same time step, which have the same
degree, will be adjacent to each other. 
All adjacencies to vertices with a larger degree are produced when the vertex is added to the graph,
and then the adjacencies to vertices with a smaller degree are made at each subsequent step. 
This results in the expression \ref{deg-growing}, i.e. $\delta=2(t-t_{c}+1)$, $t_{c}\geq1$, and we can write:

\begin{equation}\label{knn01}
k_{nn} (\delta) =\small
\frac{1}{{L_V(t_{c})\delta(t_{c},t)}}\left[{2\!\!\sum\limits_{t'_{c}=1}^{t_{c}-
1}\! {L_V (t'_{c}, )\delta(t'_{c},t)}+ 2\!\!\!\!\!\sum\limits_{t'_{c}=
t_{c} + 1}^{t}\!\!\!\!\! {L_V (t_{c})\delta(t'_{c}, t)}+ L_V
(0)\delta(0,t)} \right]
\end{equation}
where $\delta(t_{c},t)$ represents the degree of a vertex at step
$t$, which was generated at step $t_{c}$. 
Here the first sum on
the left-hand side accounts for the adjacencies made to vertices with
larger degree (i.e. $0 <t'_{c} <t_{c}$) when the vertex was
introduced at step $t_{c}$. The second sum represents the edges introduced to
vertices with a smaller degree  at each step $t'_{c}>t_{c}$. 
And the last term explains the adjacencies made to the initial vertices 
of step $0$.  Eq.~(\ref{knn01}) leads to
\begin{equation}\label{knn02}
k_{nn} (\delta) = t - t_{c}  + 2 + \frac{4}{{t - t_{c}  + 1}} -
\frac{{(t + 3)\cdot2^{1 - t_{c} } }}{{t - t_{c}  + 1}}.
\end{equation}
Writing Eq.~(\ref{knn02}) in terms of $\delta$, we have
\begin{equation}\label{knn03}
k_{nn} (\delta) = \frac{\delta}{2} + \frac{8}{\delta} + 1 -
 \frac{(t+3)\cdot2^{1+\delta/2}}{\delta\cdot 2^t}.
\end{equation}
Thus we have obtained the degree correlations for those vertices
generated at step $t_{c} \geq1$. For the initial two vertices, as each
has   degree  $ \delta_0=t+1$, we  obtain
\begin{equation}\label{knn04}
k_{nn} (\delta_0) = \frac{1}{\delta_0}\sum\limits_{t'_{c} =
0}^{t} {\delta(t'_{c}, t) = } \delta_0.
\end{equation}
From Eqs.~(\ref{knn03}) and~(\ref{knn04}), it is obvious that for
large graphs (i.e. $t\rightarrow\infty$), $k_{nn}(\delta)$ is
approximately a linear function of $\delta$, which suggests that 
Farey graphs are assortative.

Degree correlations can also be  described by the  Pearson correlation
coefficient $r$ of the degrees of the endvertices of the edges.
For a general
graph $G(V,E)$, this coefficient is defined as~\cite{Ne02}
\begin{equation}
r=\frac{|E(t)|\sum_i j_i k_i -[\sum_i \frac{1}{2} (j_i+ k_i)]^2 }{|E(t)|\sum_i \frac{1}{2} (j_i^2+ k_i^2) -[\sum_i \frac{1}{2}(j_i+ k_i)]^2 }
\end{equation}\label{rNe}
where $j_i$, $k_i$ are the degrees of the endvertices of the $i$th edge, with $i=1,\cdots ,|E(t)|$.
This coefficient is in the range
$-1\leq r \leq 1$. If the graph is uncorrelated, the correlation
coefficient equals zero. Disassortative graphs have $r<0$, while
assortative graphs have a value of $r>0$.

If $r(t)$  denotes the Pearson degree correlation of ${\cal F}(t)$, we have
%
%
\begin{proposition}
The  Pearson degree correlation coefficient of the graph ${\cal F}(t)$ is

$$
r(t )= \frac{16\cdot 2^{2t} -  2^t(4 t^3 + 6 t^2 - 40 t - 26) - (t^4 + 10 t^3 + 43 t^2 + 80 t + 2)}
{64\cdot 2^{2 t} -  2^t (6 t^3 + 18 t^2 - 6 t + 46) - (t^4 + 9 t^3 + 37 t^2 + 63 t + 18)}
$$
\label{corrFt}
\end{proposition}
\begin{proof}
We find the degrees of the endvertices for every edge of 
the Farey graph. 
Let $\langle j_i, k_i\rangle$ stand for the $i$th
edge in ${\cal F}(t)$ connecting two vertices with degrees $j_i$ and
$k_i$. Then the initial edge created at step $0$ can be expressed as
$\langle t+1, t+1 \rangle$. By construction, at step $t_{c}$
($t_{c}\geq 1$), $2^{t_{c}}$ new edges are added to the graphs.
Each new iteration will increase the degree of these vertices. These
edges connect new vertices, which have degree 2, to every vertex in
${\cal F}(t_{c}-1)$ which have the following degree distribution
at step $t_{c}-1$: $\delta(l,t_{c}-1)=2(t_{c}-l)$  ($1 \leq
l\leq t_{c}-1$) and $\delta(0,t_{c}-1)=t_{c}$. At each of the
subsequent steps of $t_{c}-1$, all these vertices will increase
their degrees by two, except the two initial vertices, the degree
of which grow by one unit.
Thus, at step $t$, the number of
edges $\langle 2(t- t_{c} + 1), 2(t- l + 1)\rangle$ ($1 \leq l\leq
t_{c}-1$) is $2^l$ and the number of edges $\langle 2(t- t_{c} +
1), t+1 \rangle$ is 2.

Now we can explicitly find the sums in Eq.~\ref{rNe}:
\begin{eqnarray*}
\sum \limits_{m=1}^{|E(t)|} j_m k_m  &= & - 79 + 5 \cdot 2^{4 + t} 
- 52 t - 15 t^2  - 2 t^3,\\
\sum \limits_{m=1}^{|E(t)|} (j_m  + k_m)  &=&  - 22 + 3 \cdot 2^{3 +
t}  - 12 t - 2 t^2,\\
\sum \limits_{m=1}^{|E(t)|}
(j_m^2 + k_m^2 ) &=&  - 206 + 13 \cdot 2^{4 + t}  - 138 t - 42 t^2  -
6t^3,
\end{eqnarray*}
which lead to the stated result.
\end{proof}

We can easily see that for $t$ large, $r(t)$ tends
to the value $\frac{1}{4}$, which again indicates that the 
Farey graphs are assortative.

\subsection{ \em Clustering coefficient}
The clustering coefficient of a graph is another parameter used to
characterize small-world networks. The clustering coefficient of a
vertex was introduced in~\cite{WaSt98} to quantify this concept:
Given a graph $G=(V,E)$, for each vertex $v\in V(G)$ with degree
$\delta_v$, its {\it clustering coefficient\/} $c(v)$ is defined as
the fraction of the ${\delta_v\choose 2}$ possible edges among the
neighbors of $v$ that are present in $G$. More precisely, if
$\epsilon_v$ is the number of edges between the $\delta_v$ vertices
adjacent to vertex $v$, its clustering coefficient is
\begin{equation}
\label{c(v)}
c(v)=\frac{2\epsilon_v}{\delta_v(\delta_v-1)},
\end{equation}
whereas the {\em clustering coefficient} of $G$, denoted by $c(G)$,
is the average of $c(v)$ over all  vertices $v$ of $G$:
\begin{equation}
\label{c(G)} c(G)=\frac{1}{|V(G)|}\sum_{v\in V(G)}c(v).
\end{equation}

Some authors, see for example \cite{NeWaSt02}, use another
definition of {\em clustering coefficient} of $G$:
\begin{equation}
\label{c'(G)} c'(G)=\frac{3\,T(G)}{\tau(G)}
\end{equation}
where $T(G)$ and $\tau (G)$ are, respectively, the number of {\it
triangles} (subgraphs isomorphic to $K_3$) and the number of {\it
triples} (subgraphs isomorphic to a path on $3$ vertices) of $G$.
A triple at a vertex $v$ is a $3$-path with central vertex $v$. Thus
the number of triples at $v$ is
\begin{equation}
\label{tau(v)}
 \tau(v)={\delta_v \choose 2}=\frac{\delta_v(\delta_v-1)}{2}.
\end{equation}
The total number of triples of $G$ is denoted by $\tau(G)=\sum_{v\in
V(G)}\tau(v)$.

Using these parameters, note that the clustering coefficient of a
vertex $v$ can also be written as $c(v)=\frac{T(v)}{\tau(v)}$, where
$T(v)={\delta_v \choose 2}$ is the number of triangles of $G$ which
contain vertex $v$. From this result, we see that $c(G)=c'(G)$ if,
and only if,
$$
|V(G)|=\frac{\sum_{v\in V(G)}\tau(v)}{\sum_{v\in
V(G)}T(v)}\sum_{v\in V(G)} \frac{T(v)}{\tau(v)}.
$$
And this is true for regular graphs or for graphs such that all vertices
have the same clustering coefficient.
$c'(G)$ is known in the context of social networks as
{\em transitivity coefficient}.

We compute here both the clustering coefficient and  the
transitivity coefficient.

\medskip
\begin{proposition}
The clustering coefficient $c({\cal F}(t))$ of the graph ${\cal F}(t)$ is
\begin{equation}
c({\cal F}(t))= \frac{1}{2^t+1}\left[  2^{t}\ln 2 -\frac{1}{2} \Phi
\left(\frac{1}{2},1,1+t\right)+\frac{4}{t+1}\right]
\end{equation}
where $\Phi$ denotes the Lerch's transcendent function.
\end{proposition}

\begin{proof}
When a new vertex $v$ joins the graph, its degree is $\delta_v=2$
and $\epsilon_v$ equals 1.
Each subsequent addition of an edge to this vertex increases both parameters by one.
Thus, $\epsilon_v$ equals  to $\delta_v-1$ for all vertices at all steps.
Therefore there is a one-to-one correspondence between the degree of a vertex and its clustering.
For a vertex $v$ of degree $\delta_v$, the  expression for
its clustering coefficient is $c(v)=2/\delta_v$.
It is interesting to note that this scaling of the clustering coefficient with the
degree has been observed empirically in several real-life networks~\cite{RaBa03}.
The same value has also been obtained in other models~\cite{OzHuOt04,ZhRoGo06,ZhRoCo06,DoGoMe02}.

We use the degree distribution found above to calculate the the clustering
coefficient of the graph ${\cal F}(t)$.
Clearly, the number of  vertices with clustering coefficient
$1$, $\frac{1}{2}$, $\frac{1}{3}$, $\cdots$, $\frac{1}{t-1}$,
$\frac{1}{t}$, $\frac{2}{t+1}$ , is equal to  $ 2^{t-1},  2^{t-2}, \ldots, 2, 1, 2$, respectively.

Thus, the clustering coefficient of the graph $c({\cal F}(t))$ is easily
obtained for any arbitrary step $t$:
\begin{eqnarray}
c({\cal F}(t))
&=&\frac{1}{ |V(t)|}\left[\sum_{i=1}^t\frac{1}{i}\cdot2^{t-i}+\frac{2}{t+1}\cdot 2\right] \nonumber \\
&=& \frac{1}{2^t+1}\left[  2^{t}\ln 2 -\frac{1}{2}
\Phi\left(\frac{1}{2},1,1+t\right)+\frac{4}{t+1}\right] .
\end{eqnarray}
\end{proof}

The clustering coefficient $c({\cal F}(t))$ tends to $ \ln2$ for large $t$.
Thus the clustering coefficient of ${\cal F}(t)$ is high.

To find the transitivity coefficient we need to calculate the number
of triangles and the number of triples of the graph.
\medskip
\begin{lemma}
The number of triangles of ${\cal F}(t)$ is $$T({\cal F}(t))=2^{t}-1, t\geq 1$$
\end{lemma}
\begin{proof}
At a given step $t$  the number of new triangles introduced to the
graph is the number of generating edges, which are the edges
introduced in the former step $L_E(t) = 2^{t-1}$. Therefore $T({\cal
F}(t))= T({\cal F}(t-1))+2^{t-1}$, and as $T({\cal F}(1))=1$, we
have the result.
\end{proof}

Moreover, again from the results found above giving the number of
vertices of each degree, we obtain straightforwardly the following
result for the number of triples:

\medskip
\begin{lemma}
The number of triples of ${\cal F}(t)$ is $$\tau({\cal F}(t))=5\cdot 2^{t+1}- t^2 -6\cdot t-10$$
\end{lemma}
$\Box$

Now the transitivity coefficient follows from the former two lemmas.
\medskip
\begin{proposition}
The transitivity coefficient of ${\cal F}(t)$ is:
$$
 c'({\cal F}(t))=\frac{3\cdot 2^t-3}{10\cdot 2^{t}-t^2-6 t-10}
$$
\end{proposition}
We see that while the clustering coefficient increases with $t$ and
tends to $ \ln2$, the transitivity coefficient tends to $0.3$.

The small-world concept describes the fact that, in many real-life
networks,  there is a relatively short distance between any pair of
nodes. In this case we will expect an average distance, and in some
cases a diameter, which scales logarithmically with the graph order.
Next we verify this two relations for the Farey graphs ${\cal
F}(t)$.

\subsection{\em Diameter}

Computing the exact diameter of ${\cal F}(t)$ can be done
analytically, and gives the result shown below.
\medskip
\begin{proposition}\label{prop:diam}
 The diameter of the graph ${\cal F}(t)$ is $diam({\cal F}(t))= t$, $t\geq 1$.
\end{proposition}
\medskip
\begin{proof}
Clearly, at steps $t = 0$ and $t = 1$, the diameter is 1.

At each step $t\geq 2$,  by the construction process, the longest
path between a pair of vertices is for some vertices added at this
step. Vertices added at a given step $t\geq 1$ are not adjacent
among them and are always connected to two vertices that were
introduced at  different former steps. Now consider two vertices
introduced at  step $t\geq 2$, say $v_{t}$ and $w_{t}$. $v_{t}$ is
adjacent to two vertices, and one of them must have been added to
the graph  at step $t-2$ or a previous one. We consider two cases:
(a) For $t = 2m$ even and from $v_{t}$ we  reach in $m$ ``jumps'' a
vertex of  the  graph ${\cal F}(0)$, which we can also reach from
$w_{t}$ in a similar way. Thus $diam({\cal F}(2m))\leq 2m$. (b) $t =
2m+1$ is odd. In this case we can stop after $m$ jumps at ${\cal
F}(1)$, for which we know that the diameter is 1, and make $m$ jumps
in a similar way to reach $w_{t}$. Thus $diam({\cal F}(2m+1))\leq
2m+1$. This bound is reached by pairs of vertices created at step
$t$. More precisely, by those two vertices $v_{t}$ and $w_{t}$ which
have the property of being connected to two vertices introduced at
steps $t-1$, $t-2$. Hence,  $diam({\cal F}(t)) = t$ for any $t\geq
1$.
\end{proof}
\smallskip

As $\log |V({\cal F}(t))| = t\cdot\ln 2$, for large $t$ we have
$diam({\cal F}(t))\sim \ln|V({\cal F}(t))|$.
\medskip

Because the graph ${\cal F}(t)$ is sparse, has a high clustering and
a small, logarithmic, diameter (and also a logarithmic average
distance, as we will prove next) our model shows small-world
characteristics~\cite{WaSt98}.

Next we find the exact analytical expression for the average
distance of the graphs ${\cal F}(t)$.

\subsection{\em Average distance}

Given a graph $G=(V,E)$ its {\em average distance} or {\em mean
distance} is defined as: $\mu (G) = \frac{1}{|V(G)|(|V(G)|-1)}
\sum_{u,v \in V(G)} d(u,v)$ where $d(u,v)$ is the distance between
vertices $u$ and $v$ of $V(G)$. 

The recursive construction of ${\cal F}(t)$ allows to obtain
the exact value of $\mu({\cal F}(t))$:
%
%
\begin{proposition}\label{theo-avg-dst}
The average distance of ${\cal F}(t)$ is
$$\mu ({\cal F}(t))= \frac{ 2^{2t}(6 t-5)+2^t(6 t+17)+5+(-1)^t  }{9\cdot 2^{2t }+9\cdot 2^t}.$$

\end{proposition}
\begin{proof}
First we find a recurrence formula to obtain transmission
coefficient $\sigma({\cal F}(t+1))$ from $\sigma({\cal F}(t))$ using
the recursive construction   of ${\cal F}(t)$ As shown in
Fig.~\ref{selfsim}, the graph ${\cal F}(t+1)$ may be obtained by
joining at three boundary vertices ($X$, $Y$, and $Z$) two copies of
 ${\cal F}(t)$ that we will  label as $F_t^{(\eta)}$ with $\eta=1,2$.
According to this construction method, the transmission
$\sigma({\cal F}(t))$ satisfies the recursive relation
\begin{equation}\label{total02}
  \sigma({\cal F}(t+1))) = 2\sigma({\cal F}(t))+ S_t,
\end{equation}
where $S_t$ denotes the sum of distances of pairs of vertices which
are not both in the same $F_t^{(\eta)}$ subgraph.

\begin{figure}
\begin{center}
\includegraphics[width=8cm]{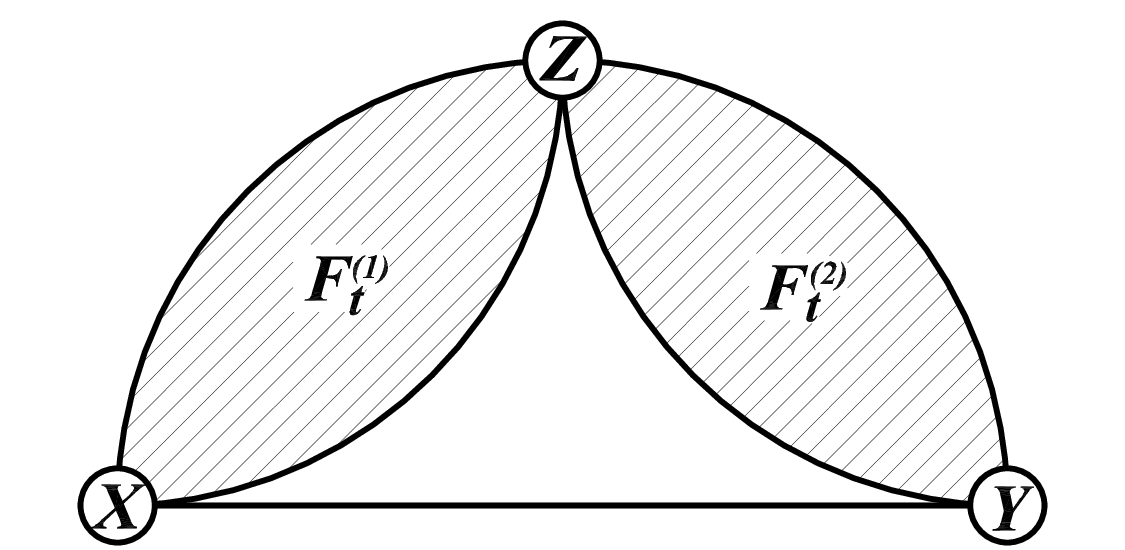}
\caption{Schematic illustration of the recursive construction of the
Farey graph. ${\cal F}(t+1)$ may be obtained by joining two copies of ${\cal F}(t)$, denoted as $F_t^{(1)}$ and $F_t^{(2)}$, which are
connected to each other as shown.} \label{selfsim}
\end{center}
\end{figure}

To calculate $S_t$, we classify the vertices in ${\cal F}(t+1)$ into
two categories: vertices $X$ and $Y$ (see figure~\ref{selfsim}) are
called connecting vertices, while the other vertices are called
interior vertices.
Thus $S_t$ 
is obtained by considering the following distances: between interior
vertices in one $F_t^{(\eta)}$ subgraph to interior vertices in the
other $F_t^{(\eta)}$ subgraph,  between a connecting vertex from one
$F_t^{(\eta)}$ subgraph  and all the interior vertices in the other
$F_t^{(\eta)}$ subgraph, and  between the two connecting vertices
$X$ and $Y$ (i.e. $d(XY)=1$).

Let us denote by $S_t^{1,2}$ the sum of  all distances between
the interior vertices, of $F_t^{(1)}$ and $F_t^{(2)}$.
Notice that $S_t^{1,2}$ does not count  paths
with endpoints at the connecting vertices $X$ and $Y$.
On the other hand, let $\Omega_t^{(\eta)}$ be the set of interior vertices in $F_{t}^{(\eta)}$.
Then the total sum $S_t$, using the symmetry
$\sum_{j \in \Omega_t^{2}}d(X,j)=\sum_{j \in \Omega_t^{1}}d(Y,j)$,
is given by
\begin{equation}\label{cross02}
S_t =S_t^{1,2} +\sum_{j \in \Omega_t^{(2)}}d(X,j)+ \sum_{j\in
\Omega_t^{(1)}}d(Y,j)+d(X,Y)=S_t^{1,2}+2\,\sum_{j \in
\Omega_t^{(2)}}d(X,j)+1.
\end{equation}
We obtain $S_t^{1,2}$ by classifying all interior vertices  of
${\cal F}(t+1)$ into three different sets according to their
distances to the two connecting vertices $X$ or $Y$. Notice that
these two vertices themselves are not included into any of the three
sets which are denoted $P_{1}$, $P_{2}$, and $P_{3}$, respectively.
This classification is shown schematically in Fig.~\ref{class}. By
construction, $d(v,X)$  and $d(v,Y)$ can differ by at most $1$ since
vertices $X$ and $Y$ are adjacent. Then the classification function
$class(v)$ of a vertex $v$ is defined as

\begin{equation}\label{classification}
class(v)=\left\{
\begin{array}{lc}
{\displaystyle{P_{1}}}  & \quad \hbox{if}\  d(v,X)< d(v,Y),\\
{\displaystyle{P_{2}}}   & \quad \hbox{if}\ d(v,X)=d(v,Y),\\
{\displaystyle{P_{3}}}   & \quad \hbox{if}\ d(v,X)>d(v,Y).\\
\end{array} \right.
\end{equation}

\begin{figure}
\begin{center}
\includegraphics[width=10cm]{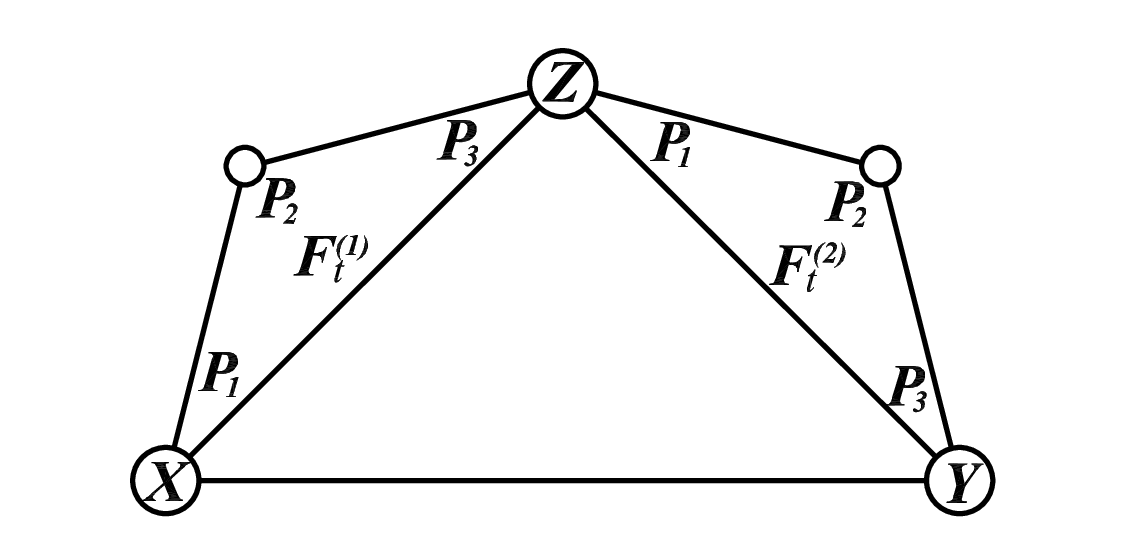}
\caption{Illustration of the vertex classification of $F_t^{(1)}$
and $F_t^{(2)}$, used to find recursively the classification of all
the interior vertices of the graph. ${\cal F}(t+1)$.}\label{class}
\end{center}
\end{figure}

This definition of vertex classification is recursive. For instance,
classes $P_{1}$ and $P_{2}$ of $F_t^{(1)}$ are in class $P_{1}$ of
$F_{t+1}^{(1)}$ and class $P_{3}$ of  $F_t^{(1)}$ is in class
$P_{2}$ of $F_{t+1}^{(1)}$. Since the two vertices $X$ and $Y$ play
a symmetrical role,  classes $P_{1}$ and $P_{3}$ are equivalent. We
denote the number of vertices in  $F_{t}$ that belong to class
$P_{i}$ ($i=1,2,3$) as $N_{t,P_{i}}$. By symmetry,
$N_{t,P_{1}}=N_{t,P_{3}}$. Therefore we have: 
\begin{equation*}
|V(t)| = N_{t,P_{1}}+N_{t,P_{2}}+N_{t,P_{3}}+2  =2\,N_{t,P_{1}} +
N_{t,P_{2}} +2.
\end{equation*}
Considering the self-similar structure of the Farey graph, we can
write the following recursive expression for $N_{t+1,P_{1}}$  and
$N_{t+1,P_{2}}$:
\begin{eqnarray*}\label{Np01}
\begin{array}{ccc}
N_{t+1,P_{1}} &=& \,N_{t,P_{1}} + \,N_{t,P_{2}}\,, \\
N_{t+1,P_{2}} &=& 2\,N_{t,P_{1}} +1\,, \\
\end{array}
\end{eqnarray*}
Together with the initial conditions $N_{1,P_{1}}=N_{1,P_{2}}=0$ we
find: 
\begin{eqnarray}\label{Np02}
\begin{array}{ccc}
N_{t,P_{1}} &=& \frac{1}{6} \left(2^{t+1}-3+(-1)^t\right), \\
N_{t,P_{2}} &=& \frac{1}{3} \left(2^t-(-1)^t\right). \\
\end{array}
\end{eqnarray}
For a vertex $v$ of ${\cal F}(t+1)$, we are also interested in the
distance from  $v$ to either of the two border connecting vertices
$X$ and $Y$ and we denote it by $\ell_v= min(d(v,X),d(v,Y))$.

Let $\ell_{t,P_{i}}  (i=1,2,3)$ denote the sum of $\ell_v$ for  all the vertices which are in class $P_{i}$ of  $F_{t}^{(\eta)}$.
Again by symmetry,  we have $\ell_{t,P_{1}}=\ell_{t,P_{3}}$, and  thus $\ell_{t,P_{1}}$, $\ell_{t,P_{2}}$
can be written recursively as follows:
\begin{eqnarray}\label{dp01}
\begin{array}{ccc}
\ell_{t+1,P_{1}} &=& \,\ell_{t,P_{1}} +\ell_{t,P_{2}}\,, \\
\ell_{t+1,P_{2}} &=& 2\,\ell_{t,P_{1}}+N_{t,P_{1}}+1\,. \\
\end{array}
\end{eqnarray}

Substituting Eq.~(\ref{Np02}) into Eq.~(\ref{dp01}), and considering
the initial conditions $\ell_{1,P_{1}}=0$ and $\ell_{1,P_{2}}=0$,
Eq.~(\ref{dp01}) can be solved and we obtain:
\begin{eqnarray}\label{dp02}
\begin{array}{ccc}
\ell_{t,P_{1}}&=& \frac{1}{81} \Big(9\cdot t\cdot(-1)^t+9\cdot t\cdot 2^t-(-1)^{-t}+2^{-t}+4 (-1)^t- \\
                      & & \qquad  -3\cdot 2^t-(-2)^{-t} (-1)^{3 t}\Big), \hfill \\
\ell_{t,P_{2}}&=& \frac{1}{81} (-2)^{-t} \Big(9 (-4)^t t+15 (-4)^t-2 (-1)^t+2 (-1)^{3t}-\\
                      & & \qquad  -2^t \left(2 (-1)^{2 t} (9 t+7)+1\right)\Big).\hfill
 \end{array}
\end{eqnarray}

After obtaining the values $N_{t,P_{i}}$ and $\ell_{t,P_{i}}$
($i=1,2,3$), we next will find $S_t^{1,2}$ and $\sum_{j \in
\Omega_t^{(2)}}d(X,j)$ expressed as a function of $N_{t,P_{i}}$ and
$\ell_{t,P_{i}}$. For convenience, we use $\Omega_t^{(\eta),i}$ to
denote the set of interior nodes belonging to class $P_i$ in
$F_{t}^{(\eta)}$. Then $S_t^{1,2}$ can be written as 
\begin{eqnarray}\label{cross04}
 S_t^{1,2} & = & \sum_{\substack{u \in \Omega_t^{(1),1},\,\,v\in F_t^{(2)}\\ v \ne Y, Z }} \! \! \! \! \! \! \! \! \! \! d(u,v) +
      \sum_{\substack{u \in \Omega_t^{(1),2},\,\,v\in  F_t^{(2)}\\ v \ne Y, Z }} \! \! \! \! \! \! \! \! \! \! d(u,v)+
      \sum_{\substack{u \in \Omega_t^{(1),3},\,\,v\in  F_t^{(2)}\\ v \ne Y, Z }} \! \! \! \! \! \! \! \! \! \! d(u,v)\,.
\end{eqnarray}
Next we calculate the first term on the right side of
Eq.~(\ref{cross04}). 
\begin{eqnarray*}\label{cross05}
 & &\sum_{\substack{u \in \Omega_t^{(1),1}, \,\,v\in F_t^{(2)}\\ v \ne Y, Z }} \! \! \! \! \! \! \! \! \! \! d(u,v)
      =  \sum_{u \in \Omega_t^{(1),1},\, v\in \Omega_t^{(2),1}\bigcup \Omega_t^{(2),2}}  \! \! \! \! \! \! \! \! \! \!  (d(u,X)+d(X,Z)+d(Z,v))\nonumber\\
      &  & + \sum_{u \in \Omega_t^{(1),1},\, v\in  \Omega_t^{(2),3}}  \! \! \! \! \! \! \! \! \! \! (d(u,X)+d(X,Y)+d(Y,v))=\nonumber \\
      \! \! \! \! \! \! \! \! \! \!   \! \! \! \! \! \! \! \! \! \!   \! \! \! \! \! \! \! \! \! \!   &=&N_{t,P_{1}}(2\ell_{t,P_{1}}+\ell_{t,P_{2}}+2N_{t,P_{1}}+N_{t,P_{2}})+\ell_{t,P_{1}}(2\,N_{t,P_{1}}+N_{t,P_{2}}).
\end{eqnarray*}
Proceeding similarly, we obtain for the second term 
\begin{eqnarray*}\label{cross06}
N_{t,P_{2}}(2\ell_{t,P_{1}}+\ell_{t,P_{2}}
+N_{t,P_{1}})+\ell_{t,P_{2}}(2\,N_{t,P_{1}} +N_{t,P_{2}}),
\end{eqnarray*}
and finally, for the third term
\begin{eqnarray*}\label{cross09}
N_{t,P_{1}}(2\ell_{t,P_{1}}+\ell_{t,P_{2}}
+N_{t,P_{1}})+\ell_{t,P_{1}}(2\,N_{t,P_{1}} +N_{t,P_{2}}).
\end{eqnarray*}
This leads to
\begin{equation}\label{cross10}
  S_t^{1,2}
  =2(2\ell_{t,P_{1}}+\ell_{t,P_{2}})(2\,N_{t,P_{1}}+N_{t,P_{2}})+N_{t,P_{1}}(3N_{t,P_{1}}+2N_{t,P_{2}}).
\end{equation}
Analogously, we find
\begin{equation}\label{cross13}
 \sum_{j \in \Omega_t^{2}}d(X,j)
  =(2\ell_{t,P_{1}}+\ell_{t,P_{2}})+(2\,N_{t,P_{1}}+N_{t,P_{2}}).
\end{equation}

Substituting Eqs.~(\ref{cross10}) and~(\ref{cross13}) into
Eq.~(\ref{cross02}) and using Eq.~(\ref{dp02}) we have the final
expression for $S_t$,
\begin{eqnarray}\label{cross14}
S_t =\frac{1}{18} \left[-5-3(-1)^t+12\cdot 2^{t}+14\cdot
2^{2t}+12\,t\cdot \,4^{t}\right].
\end{eqnarray}

With these results and recursion relations, we finally find  the
transmission coefficient. Inserting Eq.~(\ref{cross14}) into
Eq.~(\ref{total02}) and using the initial condition $\sigma({\cal
F}(0)) =1$, Eq.~(\ref{total02}) can be solved inductively,

\begin{eqnarray}\label{total04}
\sigma({\cal F}(t)) =\frac{1}{18} \left[5+(-1)^t + (6 t+17)2^t + (6
t-5) 4^{t} \right]
\end{eqnarray}
which together with the graph order leads to the stated result:
$$\mu ({\cal F}(t))= \frac{ 2^{2t}(6 t-5)+2^t(6 t+17)+5+(-1)^t  }{9\cdot 2^{2t }+9\cdot 2^t}.$$

%
\end{proof}

Notice that for a large iteration step $t$ $\mu({\cal F}(t)) \sim t
\sim \ln|V_{t}|$, which shows a logarithmic scaling of  the average
distance with the order of the graph. This logarithmic scaling of
$\mu({\cal F}(t))$ with the graph order, together with the large
clustering coefficient $c({\cal F}(t))=\ln 2$ obtained in the
preceding section, shows that the Farey graph has small-world
characteristics.

\section{Conclusion}

We have introduced and studied a family of Farey graphs, based on
Farey sequences, which are minimally 3-colorable, uniquely
Hamiltonian, maximally outerplanar, perfect, modular, have an
exponential degree hierarchy, and are also small-world.
A combination  of modularity, exponential degree distribution, and
small-world properties can be found in real networks like some social
and technical networks, in particular some electronic circuits, and
those related to several biological systems (metabolic networks,
protein interactome, etc)~\cite{Ne03,FeJaSo01}.

On the other hand, the graphs are outerplanar and many algorithms
that are NP-complete for general graphs run in polynomial time in
outerplanar graphs~\cite{BrLeSp99}. This should help to find
efficient algorithms for graph and network dynamical processes
(communication, hub location, routing, synchronization, etc).

Finally, another interesting property of this graph family is its
deterministic character which should facilitate the exact
calculation of some other graph parameters and invariants and  the
development of algorithms.


\end{document}